\documentclass[12pt]{iopart}

\usepackage{iopams}  
\usepackage{graphicx}
\usepackage{eurosym}
\usepackage{color}
\usepackage[normalem]{ulem}
\usepackage[makeroom]{cancel}
\usepackage{amsthm}

\newtheorem{theorem}{Theorem}[section]

\newcommand{\cO}{\mathcal{O}}

\newcommand{\PT}{{\cal PT}}

\newcommand{\ks}{{k_\star}} 
 
\newcommand{\ri}{\rho_\star} 
\newcommand{\rL}{{\mathrm{L}}}
\newcommand{\rR}{{\mathrm{R}}}
\newcommand{\erf}{{\textrm{erf\,}}}
\newcommand{\sech}{{\textrm{sech\,}}}

\begin{document}

\title[A universal form of complex potentials with spectral singularities]{A universal form of complex potentials with spectral singularities}

\author{Dmitry A. Zezyulin$^{1}$ \& Vladimir V. Konotop$^2$}

\address{$^1$ITMO University, St.~Petersburg 197101, Russia 
	\\
 	$^2$Departamento de F\'isica and Centro de F\'isica Te\'orica e Computacional, Faculdade de Ci\^encias, Universidade de Lisboa, Campo Grande 2, Edif\'icio C8, Lisboa 1749-016, Portugal}
\ead{dzezyulin@itmo.ru}
\vspace{10pt}

\begin{abstract}
 We establish necessary and sufficient conditions for complex potentials in the Schr\"odinger equation to enable spectral singularities (SSs) and show
   that such potentials have the universal form $U(x) = -w^2(x) - iw_x(x) + k_0^2$, where $w(x)$ is a differentiable function,   
 such that $\lim_{x\to \pm \infty}w(x)=\mp k_0$,  and $k_0$ is a nonzero real.
  We also find  that when $k_0$ is a complex number, then   
 the eigenvalue of the corresponding Shr\"odinger operator has an exact solution which, depending on $k_0$, represents a coherent perfect absorber (CPA), laser, a localized bound state, a quasi bound state  in the continuum (a quasi-BIC), or an exceptional point (the latter requiring additional conditions). Thus, $k_0$ is a bifurcation parameter that 
 describes transformations among all those solutions.  Additionally, in a more specific case of a real-valued function $w(x)$ the resulting  potential, although not being $\PT$ symmetric,  can feature a self-dual spectral singularity associated with the CPA-laser operation.  In the space of the system parameters, the transition through each self-dual spectral singularity corresponds to a bifurcation of a pair of complex-conjugate propagation constants from the continuum. The bifurcation of a first complex-conjugate pair corresponds to the phase transition from purely real to complex spectrum.
\end{abstract}

%
\vspace{2pc}

%
%
%

\section{Introduction}

Singularities of the spectral characteristics of non-Hermitian operators, alias spectral singularities (SSs), were introduced in mathematical literature more than six decades ago~\cite{SS} and were well studied since then  \cite{SS1,SS2,Guseinov}. Independently on these studies, there were appearing physical examples of absorbers \cite{Khapalyuk,Khapalyuk2,Poladian} (see also \cite{Rosanov}) and lasers~\cite{Poladian} of systems possessing SSs (although without direct reference on the notion of SS). The close relation between SSs, and the physical concepts of coherent perfect absorber (CPA),   laser, and 
zero width resonances were established more recently in a series of theoretical~\cite{Most09,MostMeDe,Longhi10,Chong10} and experimental~\cite{Wan} works (see e.g. \cite{Baranov} for review of physical realizations and applications of CPAs).

While the definition of a SS can be formulated in terms of poles of a truncated resolvent of a non-Hermitian Schr\"odinger operator in any spatial dimension (see e.g. the discussion in \cite{KLV2019} and references therein), in this work we deal only with one-dimensional setting. In this case a convenient description of the SSs can be elaborated in terms of the transfer matrix $M(k)$ depending on the wavenumber $k$, when real zeros of the matrix element $M_{22}(k)$  determine SSs~\cite{Most09,MostMeDe}. This approach indicates that 
for a given localized complex potential the existence of a SS is a delicate property, requiring precise matching of the physical parameters ensuring a real zero of a complex function $M_{22}(k)$.  
Hence a number of free parameters of the potential must be big enough in order to ensure the existence of a SS. In the meantime, it turns out that in spite of these, sometimes severe, constraints, a number of potentials that support SSs with the prescribed properties (these including the wavelength, the order of a SS, the number of SSs, etc.)  is infinitely large and can be constructed in an algorithmic way~\cite{Most14a,Most14b,KLV2019}. The remaining questions, however, are whether the structure of all such complex potentials, supporting SSs,  and whether the field structure corresponding to those potentials  have something in common or not. 

In this paper we give positive answer to both above questions. More specifically, subject to some  quite weak (physically) constraints, we establish necessary (section~\ref{sec:necessary}) and sufficient (section~\ref{sec:sufficient}) conditions on the form of a complex potential   that features SSs.  This universal form reads   $U(x) = -w^2(x) - iw_x(x) + k_0^2$, where $U(x)$ is the complex potential,  $k_0$ is a nonzero real, and $w(x)$ is a  differentiable function which features asymptotic behaviour $\lim_{x\to \pm \infty}w(x)=\mp k_0$ and   satisfies certain additional (not very restrictive) requirements. Our proofs are based on a specific representation of the SS solutions which allows for describing the parametric transformation of the complex potential making it supporting other types of solutions including bound states, quasi-bound states in continuum, and exceptional point solutions (\ref{sec:pot_func_k}).  Furthermore, in a more specific situation of real-valued function $w(x)$ we demonstrate that the found SS solution can coexist with another, self-dual  spectral singularity which corresponds to the combined CPA-laser operation (section~\ref{sec:exmaples}). In concluding section~\ref{sec:concl} we provide an outlook towards possible practical implementation of our findings and   some promising generalizations of the presented theory.
%
 
\section{Universal form of a complex potential resulting in a spectral singularity} 
\label{sec:necessary}
Our main goal is to study some general properties of special solutions of a one-dimensional Schr\"odinger equation
 \begin{eqnarray}
 \label{SE} 
 -\psi^{\prime\prime}+U(x)\psi= k^2\psi, 
 \end{eqnarray}
 where $U(x)$ is a spatially localized complex-valued potential, i.e.
 \begin{eqnarray}
 \label{asympt_U}
  \lim_{x\to\pm\infty }U(x)=0,
 \end{eqnarray}
 and $k$ is a spectral parameter (hereafter a prime stands for a derivative with respect to $x$).  
 More specifically, we are interested in solutions $\psi_0(x)$  of (\ref{SE}) which satisfy the conditions as follows:
 \begin{itemize}
 	\item[(i)]  the function $\psi_0(x)$ is a continuously differentiable, i.e.,  $\psi_0\in C^1(\mathbb{R})$;
 	
 	\item[(ii)] for a given value of the spectral parameter   $k=k_0\in\mathbb{R}$, the function  $\psi_0(x)$ is characterized by the asymptotic behaviour
 	\begin{equation}
 	\label{eq:asnes}
 	\lim_{x\to\pm\infty} [\psi_0^{(n)}(x) - (\pm ik_0)^ne^{\pm ik_0x}\rho_\pm]=0,  \quad \rho_{\pm}\ne 0, \quad n=0,1,2,
 	\end{equation} 
 \end{itemize}
where  the superscript ``${(n)}$'' denotes the $n$th derivative in $x$:  $\psi_0^{(n)}(x)\equiv d^n\psi_0/dx^n$,    and $\rho_\pm\in\mathbb{C}$ are nonzero complex constants.  A solution satisfying the formulated constraints will be referred to as a SS-solution. Requirements  (\ref{eq:asnes}) imply that a SS-solution $\psi_0(x)$ describes a laser for $k_0>0$ and a CPA for $k_0<0$. Taking into that the $2\times 2$ transfer matrix $M(k)$ links the Jost solutions of any localized potential (see e.g.~\cite{Sabatier}), the definition of SS-solution given above is consistent with the standard definition of a SS. Namely, $k_0$ is a SS if $M_{22}(k_0)=0$ (hereafter $M_{ij}(k)$ with $i,j=1,2$ stand for the respective elements of $M(k)$).

Consider some   function $\psi_0(x)$ that has  the above properties (i) and (ii) and  solves equation  (\ref{SE}). Then one can show that subject to some additional conditions such $\psi_0(x)$ requires the potential $U(x)$ to have the asymptotic behavior   (\ref{asympt_U}), as well as to admit a special representation, given by  equation (\ref{Wadati}) below. Starting with  these additional conditions, we assume that $\psi_0(x)$ has only a finite number $N$ of  simple roots  (i.e. roots of multiplicity one) denoted by $x_j$ ($j=1,2,..,N$) and ordered as
\begin{equation}
\label{eq:order}
-\infty < x_1<x_2<\ldots <x_N<\infty .
\end{equation}
Then  it follows from Taylor's theorem that at $x\to x_j$
\begin{eqnarray}
\label{zeros}
\begin{array}{l}
\psi_0(x) = \psi_{0}^{\prime}(x_j)(x-x_j) +o(x-x_j) ,
\quad
\psi_{0}^\prime(x) = \psi_{0}^{\prime}(x_j)  + o(1), 
\\[3mm]
\psi_{0}^{\prime}(x_j) \ne 0.
\end{array}
\end{eqnarray}

For the next consideration, it is convenient to introduce  intervals on the real axis where $\psi_0(x)\neq 0$ 
(see   figure~\ref{fig:proof})
\begin{equation}
\label{eq:III}
I_0 = (-\infty, x_1), \ I_1=(x_1, x_2), \ \ldots \ I_{N-1}=(x_{N-1}, x_N),\ I_N = (x_N, \infty),
\end{equation}
as well as the set $I$ of all point  of the real axis where $\psi_0(x)$  is nonzero: 
\begin{equation}
\label{set_I}
I=I_0\cup I_1\cup\cdots\cup I_N.
\end{equation}
 \begin{figure}
 	\centering
 	\includegraphics[width=0.8\linewidth]{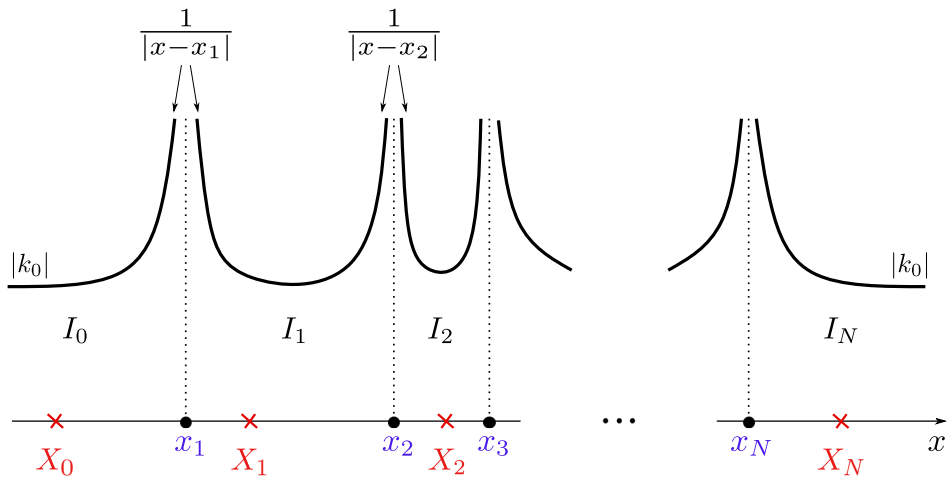}
 	\caption{Illustration to theorems~\ref{theor:necessary} and \ref{theor:sufficient}. Points $x_1<x_2<\ldots< x_N$ divide the real $x$-axis into intervals  $I_0, I_1, \ldots, I_N$. Each interval contains an arbitrarily chosen point $X_0, X_1, \ldots, X_N$. Additionally, a schematics of  singularities and asymptotic behaviour  of   $|w(x)|$ are indicated.}
 	\label{fig:proof}
 \end{figure}	
   
For $x\in I$ one can define the function  
\begin{eqnarray}
\label{w}
w(x) = i\frac{\psi_{0}^{\prime}(x)}{\psi_0(x)}
\end{eqnarray}
which  is continuous in $I$ and has the following additional properties (see   figure~\ref{fig:proof}) 
\begin{eqnarray}
\label{asympt_w}
\lim_{x\to \pm \infty}w(x)=\mp k_0,
\\
\label{singularity_w}
w(x) = \frac{i}{x-x_j} +o\left(\frac{1}{x-x_j} \right)\quad \mbox{as\ } x\to x_j. 
\end{eqnarray}
Additionally,  let us also  assume that  in each interval $I_j$ there exists  the second  derivative $\psi_{0}''(x)$ which is bounded     
in each closed subset of $I$. Then $w(x)$ is differentiable in $I$,  and  its derivative $w'(x)$ is also bounded in each closed subset of $I$. Asymptotic behavior of the SS-solution $\psi_0$ required in (\ref{eq:asnes}) implies that 
\begin{eqnarray}
 \lim_{x\to \pm \infty}w^\prime (x)=0.
\end{eqnarray} 
Then, for $\psi_0(x)$ to be a solution of Schr\"odinger equation (\ref{SE}), the complex potential in this equation must have the form
\begin{equation}
\label{Wadati}
U(x) =  
- w^2(x) - iw^\prime(x)+k_0^2.
\end{equation}
For $k_0=0$ and real-valued function $w(x)$, the potentials of this type  have been  proposed in earlier works~\cite{Andrianov,Wadati08} for obtaining complex potentials with real spectra. It follows from (\ref{eq:asnes}) that  the  potential introduced through (\ref{Wadati}) has zero asymptotic behaviour, as prescribed by (\ref{asympt_U}), and is a bounded function in any closed interval belonging to $I$.


Since the starting point of the above analysis was a SS-solution, now we can formulate a necessary condition for a complex potential $U(x)$ to result in a SS of the respective Schr\"odinger operator. 
 \begin{theorem}[Necessary condition]
 	\label{theor:necessary}
 	Let a function  $\psi_0 \in C^1(\mathbb{R})$  
 	have at most a finite number of simple roots $x_j$ ordered as in (\ref{eq:order}),
 	have asymptotic behaviour (\ref{eq:asnes}), and in the set $I$ defined by (\ref{set_I}) have the second derivative $\psi_0''(x)$ which is bounded in each closed subset of $I$. If   $\psi_0(x)$ solves Schr\"odinger equation~(\ref{SE})  with a complex potential $U(x)$, then  (i) 
 	$U(x)$ is bounded in any closed sub-interval of $I$ and vanishes at $|x|\to\infty$,
 	(ii) for $x\in I$ this potential allows for representation (\ref{Wadati}) with the function $w(x)$ having asymptotics (\ref{asympt_w}), and (iii) the function $w(x)$ is defined by (\ref{w}).   
 \end{theorem}

We notice that in the isolated points  $x_j\notin I$  the potential $U(x)$ is allowed to have singularities. If however, a function $\psi_{0}$ is different from zero on the whole real axis, i.e.  $\psi(x)\neq 0$ for $x\in\mathbb{R}$, then neither $w(x)$ nor $U(x)$ have singularities. In this last case, solving (\ref{w}) with respect to $\psi_0$, we find  that the solution corresponding to the spectral singularity of potential (\ref{Wadati}) 
  can be expressed directly through the function $w(x)$, as 
  \begin{equation}
  \label{eq:sol_SS}
  \psi_0(x) = \rho \exp\left[ -i \int_{x_0}^x w(\xi)d\xi \right], 
  \end{equation}
where $x_0$ and $\rho$ are  arbitrary constants ($\rho\ne 0$). Recently, some SS-solutions having the from (\ref{eq:sol_SS}) and related to the potential (\ref{Wadati}) were discussed in \cite{Horsley}.
  
The potential $U(x)$ determined in Theorem~\ref{theor:necessary} may have discontinuities. For the sake of illustration, let us show that the known example of a complex rectangular potential~\cite{Most09,KLV2019}: 
\begin{eqnarray}
\label{rect}
U(x)=\left\{
\begin{array}{ll}
k_0^2-\kappa^{2}, & \mbox{at $|x|\leq 1$}
\\
0 & \mbox{at $|x|>1$}
\end{array}
\right. , 
\end{eqnarray}
where $\kappa$ is a complex number, for the parameters enabling a SS-solution can be represented in the from (\ref{Wadati}).
Indeed,  let $k_0$ be a SS. This means that either $\kappa$ is such that for a real $k_0$ we have $k_0=i\kappa\tan\kappa$, and then we  compute the an SS-solution
\begin{equation}
\psi_0(x)=  \left\{ 
\begin{array}{ll}
\kappa  e^{ik_0(|x|-1)} & |x| \geq 1
\\
\kappa\cos(\kappa x)/\cos \kappa  & |x| \leq 1
\end{array}
\right.      
\end{equation}
and respectively
\begin{eqnarray}
w(x)=\left\{  
\begin{array}{ll}
\mp k_0 & \pm x>1
\\
{ -i\kappa\tan(\kappa x)}  & |x| \leq 1
\end{array}
\right.,
\end{eqnarray}
or $k_0=-i\kappa\cot\kappa$ and we obtain an odd SS-solution
\begin{equation}
\label{examp:odd_solut}
\psi_0(x)=  \left\{ 
\begin{array}{ll}
{\pm}\kappa  e^{ ik_0(|x|-1)} & \pm x > 1
\\
\kappa\sin(\kappa x)/\sin \kappa  & |x| \leq 1
\end{array}
\right.  
\end{equation}
with the respective $w(x)$ function
\begin{eqnarray}
\label{examp:odd_w}
w(x)=\left\{  
\begin{array}{ll}
\mp k_0 & \pm x>1
\\
{i\kappa\cot(\kappa x)} & |x| \leq 1
\end{array}
\right. .
\end{eqnarray}
Since the odd solution has zero at $x=0$, the function $w(x)$ has a singularity at this point, which however correspond to a bounded value $U(0)$ of the potential (as this is schematically illustrated in Fig.~\ref{fig:proof}).

\section{Construction of solutions corresponding to spectral singularities. Sufficient condition.}
 \label{sec:sufficient}
 
In the previous section, we have demonstrated that if Schr\"odinger equation (\ref{SE}) has a  SS solution $\psi_0(x)$ obeying certain additional properties, then the potential $U(x)$ in this equation admits a representation in the form  (\ref{Wadati}), where function $w(x)$ can be found from the solution $\psi_0(x)$. In this section we address a converse situation. Suppose that function $w(x)$ is known, and the potential $U(x)$ has the form (\ref{Wadati}). Our goal now is to construct  a SS-solution for such a potential. If such a solution can be found, then the representation (\ref{Wadati}), eventually supplemented with    additional constraints, is also a sufficient condition for the existence of SSs.
   
The starting point of this analysis will be the simple solution (\ref{eq:sol_SS}) which defines a zero-free SS solution supported by a continuous function $w(x)$. In fact, this solution can be generalized  on the    case  when $w(x)$  has discontinuities of a certain type.  As we shall demonstrate below, in this situation   solution $\psi_0(x)$ should be defined piecewise on  each continuity interval of $w(x)$, and the amplitude of $\psi_0(x)$ vanishes at the points of discontinuity.  

Turning to rigorous formulation, we consider a function $w(x)$ which is defined piecewise in intervals $I_j$ ($j=0,1,\ldots N$) introduced in   (\ref{eq:III}) (we notice that $N=0$ corresponds to a continuous function $w(x)$). More specifically, we assume that    (i) for $x\in I_j$, $w(x) \equiv  w_j(x)$, where each function $w_j(x)$ is   continuous and    piecewise differentiable in $I_j$; (ii) in the points $x_j$ the function  $w(x)$ has discontinuities defined by the representations  
 \begin{eqnarray}
 \label{sing_w_gen}
 \begin{array}{l}
  \displaystyle{ w_0(x)  
 	= \frac{i}{x-x_1} + w_1^-(x),}
 \\[4mm]
 	\displaystyle{w_1(x) = \frac{i}{x-x_1} + w_1^+(x) = \frac{i}{x-x_2}+w_2^-(x),}
 	\\
 	\vdots 
 	\\
 	\displaystyle{w_{N-1}(x) = \frac{i}{x-x_{N-1}} + w_{N-1}^+(x) = \frac{i}{x-x_N}+w_N^-(x),}
 	\\[4mm]
 	\displaystyle{w_N(x) 
 	= \frac{i}{x-x_N} + w_N^+(x),}
 	\end{array}
 \end{eqnarray}
 where $w_j^\pm = \cO(1)$ at ${x\to x_j\pm 0}$; and (iii) has the asymptotic behaviour
\begin{eqnarray}
\label{asimp_w_infty}
\begin{array}{ll}
 w_0(x) =k_0 + \tilde{w}_{-\infty}(x), & \tilde{w}_{-\infty}=\cO(|x|^{-2}) \quad\mbox{at $x\to-\infty$},
 \\[2mm]
   w_N(x) 
 = -k_0 +  \tilde{w}_{\infty}(x), & \tilde{w}_{\infty}=\cO(|x|^{-2}) \quad\mbox{at $x\to\infty$}.
\end{array}
\end{eqnarray}
Behaviour of $w(x)$ is illustrated schematically in figure~\ref{fig:proof}.
We also notice that the condition (\ref{sing_w_gen}) defining the singularities is more restrictive than condition  (\ref{singularity_w}) considered previously.

In each interval $I_j$ we choose  an arbitrary point  $X_j\in I_j$, and define a function $\phi_j(x)$  
 \begin{eqnarray}
 \label{psi_j}
  \phi_j(x) :=  \rho_j \exp\left[-i\int_{X_j}^x w_j(\xi)d\xi\right], \qquad x\in I_j,
 \end{eqnarray}
 where $\rho_j$ is a complex constant undefined, so far. 
   One can readily verify  that for the introduced functions $\phi_j(x)$ ($j=1,2,\ldots,N$) the following limits are valid
  \begin{eqnarray}
  \label{limit_x_j}
   \lim_{x\to x_j+0}\phi_j(x)=\lim_{x\to x_j-0}\phi_{j-1}(x)=0.
  \end{eqnarray}
  Indeed, for the left limits in points $x_j$, we compute
  \begin{equation*}
  \lim_{x\to x_j-0}\phi_{j-1}(x) = \rho_{j-1} \lim_{x\to x_j-0}\exp\int_{X_{j-1}}^{x_j}\left[\frac{1}{\xi-x_j}-iw_j^-(\xi)\right]d\xi=0.
  \end{equation*}
  The right limits are verified similarly.
 
 One can also verify the validity of the following limits:
   \begin{eqnarray}
  \label{limit_infinities}
 \lim_{x\to-\infty} [\phi_0(x) - e^{-ik_0x}\rho_-] =  \lim_{x\to\infty} [\phi_N(x) - e^{ik_0x}\rho_+]=0,
  \end{eqnarray}
  where $\rho_\pm$ are   constants given by
  \begin{eqnarray}
   \label{rho_m}
   \rho_- = \rho_0 \exp\left[ik_0X_0 -i\int_{X_0}^{-\infty} \tilde{w}_{-\infty}(\xi)  d\xi\right],
   \\
   \label{rho_p}
   \rho_+ = \rho_N \exp\left[-ik_0X_N -i\int_{X_N}^{\infty} \tilde{w}_{\infty}(\xi)  d\xi\right].
  \end{eqnarray}
  Note that the improper integrals in these formulas converge due to requirements (\ref{asimp_w_infty}).
   

Let us now define the function
\begin{eqnarray}
\label{def_psi_suff}
\psi_0(x):=\left\{
	\begin{array}{ll}
	\phi_j(x) & x\in I_j
	\\
	0 & x\in\{x_1, x_2, \ldots, x_N\}
	\end{array}
	\right. .
\end{eqnarray}
By the above properties of $\phi_j(x)$ the so defined function $\psi_0(x)$ is continuous on the whole real axis and possesses asymptotic behaviour required for  a SS-solution. Hence, for  such a function to represent  a SS-solution, it must be also continuously differentiable on the whole real axis. In the following theorem we prove that this goal can be achieved  by  the  proper choice of 
the parameters $\rho_j$.

 \begin{theorem}[Sufficient condition]
 	\label{theor:sufficient}
 	 Let $X_j\in I_j$ for $j=0,\ldots,N$ and $\rho_0$ be an arbitrary nonzero constant. Define amplitudes $\{\rho_1,\ldots,\rho_N\}$  by the recurrent law as follows ($j=1, 2, \ldots, N$):
 	\begin{eqnarray}
 	\rho_{j} =  \rho_{j-1}\frac{X_j-x_j}{X_{j-1}-x_j} 
 	\exp\left[i\int_{X_j}^{x_j} w_j^+(\xi)d\xi - i\int_{X_{j-1}}^{x_j} w_j^-(\xi)d\xi  \right]. 
 	\label{eq:reccur}
 	\end{eqnarray}
 	Then for the function $\psi_0(x)$ defined by (\ref{def_psi_suff}) the following properties hold:
 	\begin{itemize}
 		\item[(a)]  $\psi_0(x) \in C^1(\mathbb{R})$; 
 		\item[(b)] there exist nonzero constants $\rho_\pm$ such that 
 		\begin{equation}
\lim_{x\to\pm\infty}[\psi_0(x) - \rho_\pm e^{\pm ik_0x}] = 0;
 		\end{equation}
 		\item[(c)] 	for $x\in I_j$,   $\psi_0(x)$ solves the equation  
 		\begin{equation}
 		\label{eq:psi_j}
 			-\psi_{0}''  +  U_j(x) \psi_0 = k_0^2\psi_0
 		\end{equation}
 		where   
 		\begin{equation}
 		U_j(x) = -w^2_j(x) - iw_{j}'(x) + k_0^2.
 		\end{equation}
 	\end{itemize}		
 \end{theorem}

 \begin{proof} Bearing in mind already established properties of the functions $\phi_j(x)$, for the proof of this theorem we only need to verify the continuity of the derivative $\psi_0^\prime$ in $\mathbb{R}$. For $x\in I$ this follows from the definition (\ref{def_psi_suff}), because in each interval $I_j$ we compute
 	\begin{equation}
\psi_0'(x) = -iw_j(x)\phi_j(x), \quad x\in I_j.
 	\end{equation}
 	The $\psi_0'(x)$ is continuous in $I_j$  due to the continuity of $w_j(x)$ and $\phi_j(x)$. Hence, it remains to prove that $\psi_0'(x)$ is continuous in the points $x_j$, $j=1,\ldots, N$. To this end, for each point $x_j$, we compute the left and right derivatives of $\psi_0(x)$:
 	\begin{eqnarray*}
%
 	\psi_0'(x_j-0) 
= -\lim_{\epsilon\to +0}\frac{\psi_0(x_j-\epsilon)}{\epsilon} = \frac{\rho_{j-1}}{X_{j-1}-x_j} \exp\left[-i\int_{X_{j-1}}^{x_j}w_j^-(\xi)d\xi\right]
 	\end{eqnarray*}
 	and
 	\begin{eqnarray*}
 	\psi_0'(x_j+0)
 	= \lim_{\epsilon\to+0}\frac{\psi_0(x_j+\epsilon)}{\epsilon} = \frac{\rho_{j}}{X_{j}-x_j} \exp\left[-i\int_{X_{j}}^{x_j}w_j^+(\xi)d\xi\right].
 	\end{eqnarray*}
 	Thus, if the amplitudes $\rho_{j-1}$ and $\rho_j$ are connected through the relation (\ref{eq:reccur}), then  the derivative of $\psi_0(x)$ is continuous in $\mathbb{R}$.
  \end{proof}
 

  \section{From a spectral singularity to exceptional point and to a bound state in continuum
 }
 \label{sec:pot_func_k}

 As we have demonstrated above, the complex potential $U(x)$ of the form  (\ref{Wadati}) always has a SS solution, provided that the function $w(x)$ satisfies certain conditions. The most important of those conditions,  which can be expressed by equation (\ref{asympt_w}), requires the function $w(x)$ to approach constant values $\mp k_0$ as $x\to\pm \infty$, respectively. In those considerations, it was important that the parameter $k_0$ was real. 
%
Now we relax this requirement and consider physically relevant solutions with a complex $k_0$. In other words, $k_0$ will be considered as a parameter that modifies   the potential in order to obtain a prescribed solution. We note that a similar problem for $\PT-$symmetric potentials was recently addressed in \cite{Ahmed2018} based on the fact that a SS can be viewed as a complex discrete energy corresponding to a solution with outgoing plane wave asymptotics.

For the sake of simplicity, we assume that $w(x)$ is  a continuous function, which allows to use solution  (\ref{eq:sol_SS}) as a starting point. Generalization for discontinuous functions $w(x)$  can be developed straightforwardly following the ideas of section~\ref{sec:sufficient}; it is not presented here. 

Consider a complex-valued continuous function  $w(x)$  which tends to $\mp k_0$ as $x\to\pm \infty$, where $k_0$ is an arbitrary \emph{complex} constant. Then the function    $\psi_0(x)$, defined formally by equation  (\ref{eq:sol_SS}),  solves Schr\"odinger equation (\ref{SE}) with the potential given by (\ref{Wadati}).
Such a solution, however, may be   irrelevant if it does not satisfy physically meaningful boundary conditions (i.e. conditions at $x\to\pm\infty$). 
The asymptotic behavior of $w(x)$ implies that at large $x$  the solution (\ref{eq:sol_SS}) behaves as  $\psi_0 \sim \rho e^{\pm ik_0 x}$   as $x\to \pm \infty$.
 Thus, solutions characterized by amplitudes growing with $|x|$, and hence physically irrelevant correspond to $k_0$ in the lower complex half-plane, i.e., to Im$k_0<0$. Such solutions will not be considered below. For other values of the complex $k_0$  one can distinguish several cases (see the diagram in Fig.~\ref{fig:diag}). 
 \begin{figure}
 	\centering
 	\includegraphics[width=0.8\columnwidth]{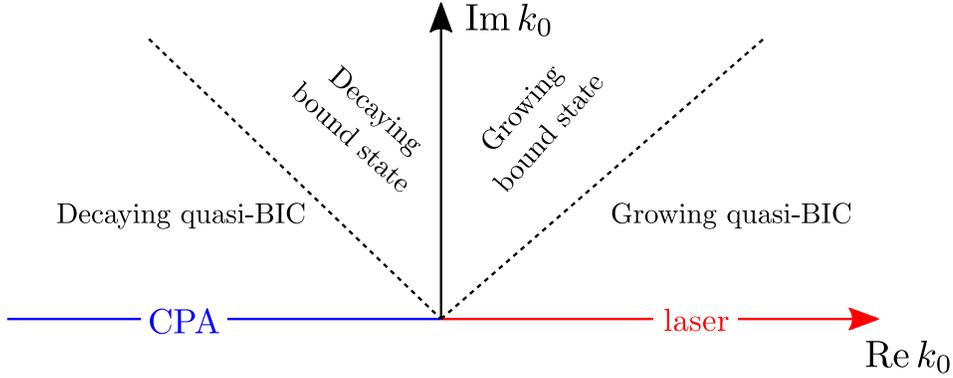}
 	\caption{Diagram on the complex $k_0$-plane  showing different features of the exact solution (\ref{eq:sol_SS}).  }
 	\label{fig:diag}
 \end{figure}

If  $k_0$ is nonzero real (i.e., Im$\,k_0=0$), then the solution (\ref{eq:sol_SS}) describes:
 \begin{itemize} 
 	\item{\em A spectral singularity}   
 	\end{itemize} 
This is the case considered in the previous sections. For positive and negative $k_0$, this solution is laser or  a CPA, respectively.
 
 If  $k_0$ is in the upper complex half-plane (i.e., Im$\,k_0>0$) then the solution (\ref{eq:sol_SS}) describes  a bound state, i.e. satisfies the localization condition
 \begin{eqnarray}
 \lim_{x\to\infty}\psi_0(x) = \lim_{x\to-\infty}\psi_0(x)=0.
 \end{eqnarray}
 Meantime, the nature of such a bound state  can be different and depends on the associated   eigenvalue $k_0^2$:
 \begin{itemize} 
 	
 \item{\em A bound state} occurs if  Re$\, k_0^2<0$, i.e.,  $\arg k_0 \in (\pi/4, 3\pi/4)$, and
	$\displaystyle{\int_{-\infty}^{\infty} \psi_0^2dx\neq 0}$. The respective bound state can be 
	\begin{itemize}
		\item{\em stationary}, if  Re$\, k_0= 0$,
		\item{\em growing} if   Re$\, k_0>  0$,  
		\item{\em decaying} if   Re$\, k_0 < 0$. 
	\end{itemize}
All these bound states are characterized by the real part the eigenvalue $k_0^2$ outside of continuum. 

\item{\em A bound state  with the real part in continuum ({\em quasi-}BIC)} takes place if the real part of eigenvalue $k_0^2$ is positive, i.e. lies in the continuous spectrum. In terms of the complex parameter $k_0$, this correspond to $\arg k_0 \in (0, \pi/4)\cup (3\pi/4, \pi)$. For previous discussion of solutions of these types see e.g.~\cite{BIC_optics1,BIC_optics2}.

\item{\em Exceptional point (EP) }~\cite{Kato} is found if the  quasi-self-orthogonality condition holds, i.e.,  
$\displaystyle{\int_{-\infty}^{\infty} \psi_0^2dx= 0}$~\cite{Moiseev}. An EP-solution can also be stationary, growing or decaying, what is determined by the real part of $k_0$ as specified above. It also can be either a bound state with a real part out of continuum or a quasi-BIC (i.e., having the real part in the continuum.)


			
 	
 \end{itemize} 
 Notice that the above classification of growing and decaying bound states stems from the physical interpretation of the stationary Schr\"odinger equation~(\ref{SE}) as a reduced form of the (dimensionless) wave equation
 \begin{eqnarray}
 \label{wave}
- \frac{\partial^2\psi}{\partial x^2}+U(x)\psi=-\frac{\partial^2 \psi}{\partial t^2}
 \end{eqnarray}	
 after the ansatz $\psi\propto e^{-i\omega t}$, with $\omega^2=k^2$ and Re$\,\omega>0$, when growing with time solution corresponds to Im$\,\omega>0$; or alternatively as the optical parabolic approximation (or equivalently as time dependent Schr\"odinger equation with $t$ replaced by dimensionless propagation distance $z$)
 \begin{eqnarray}
 \label{parabolic}
 - \frac{\partial^2\psi}{\partial x^2}+U(x)\psi=i\frac{\partial \psi}{\partial z}
 \end{eqnarray}	
 after the ansatz $\psi\propto e^{ibz}$, where the propagation constant can be computed as $b=-k_0^2$. Respectively,   bound states    have  Re$\,b>0$, and    growing (decaying) with $z$  solution corresponds to Im$\,b<0$ (Im$\,b>0$).
 
 Thus considering  a family of  potentials of the form (\ref{Wadati}), where $w(x)$ is fixed and $k_0$ is changing in the complex plane, one   obtains a specific complex potential having solutions in a form of a SS or in a form of  bound state with desired properties. The only exception of this rule is the parametric transition between a bound state and an EP-solution (where two bound states coalesce), because for such a transition to occur the form of $w(x)$ itself might need to be changed, and hence an additional free parameter may be needed. As a simple example illustrating this parametric dependence we consider $k_0= i$ and the complex-valued function $w(x)$ defined as
 \begin{equation}
 \label{example}
 w(x;z) = -i\tanh x +  {z}\sech x,
 \end{equation}
 where $z$ is an additional complex parameter. 
 Using  the explicit expression  (\ref{eq:sol_SS}), with $x_0=0$,  we obtain the exact bound state  solution
 \begin{equation}
 \label{eq:sech}
 \psi_0(x;z) =  \exp\left\{\frac{iz}{2}\left({\pi}  - 4\arctan(\exp x)\right)\right\}\textrm{sech\,} x.
 \end{equation}
 Now it is straightforward to compute  
 \begin{equation}
 \int_{-\infty}^{\infty}\psi_0^2(x;z) dx =  \frac{2\cos(\pi z)}{1-4z^2}.
 \end{equation}
 This integral becomes zero at   $z=z_n =\pm(n+1/2)$ and $n=1,2,\ldots$. Thus, $\psi_0(x; z_n)$ given by (\ref{eq:sech}) corresponds to the EP of Eq.~(\ref{SE})  with the potential $U(x)$ generated by $w(x; z_n)$.  The resulting complex potential computed according to (\ref{Wadati}) reads 
 \begin{equation}
  U(x)= -(z^2+2)\sech^2x + 3iz\tanh x\, \sech x.
 \end{equation}
For real values of the parameter $z$ the obtained potential  is $\PT$-symmetric, i.e.  satisfies the property $U(x)=U^*(-x)$ and  belongs to the family of $\PT$-symmetric Scarff~II potentials~\cite{Scarf-2,Scarf1,Scarf2}. 

\section{Self-dual spectral singularities for  complex potentials (\ref{Wadati}) with real-valued $w(x)$. Examples of SS-solutions.}  
\label{sec:exmaples} 
 
\subsection{Transfer matrix approach and basic properties}
\label{sec:scatt_matrix}

Let us now consider the obtained results on the existence of a SS solution and specific form of the potential in more conventional terms of the zeros of the transfer matrix. For this sake it will be convenient to rewrite the scattering problem (\ref{SE}) using the notation for the Hamiltonian operator $H$:
  \begin{equation}
 \label{eq:Sch}
 H\psi = k^2 \psi, \qquad 
 H = -\partial_{xx} + U(x)
 \end{equation}
 and rewrite the Schr\"odinger equation (\ref{SE}) in the form of eigenvalue problem $H\psi = k^2\psi$. 
 
 In this section, we consider only real valued functions $w(x)$ which, as above, satisfy the boundary conditions $w(x)\to\mp k_0$ as $x\to\pm\infty$. Respectively, the constant $k_0$ is also real in this section. In this case  the Hamiltonian $H$ admits an important additional symmetry which is expressed by the relation \cite{NY16}
 \begin{equation}
 \label{eq:pseudo}
 \eta H = H^*\eta, \quad \eta = \partial_x  + iw(x),
 \end{equation}
 which is close to the property of the pseudo-Hermiticity \cite{Wick,Mostafazadeh2002,Solombrino}.
 This implies that for any solution $\psi$ with a real   $k$ one can construct another solution with the same $k$ in the form $(\eta \psi)^*$.  
 
 For the scattering problem  (\ref{eq:Sch})  
 one can  introduce a pair of left (superscript ``$\rL$'') and right (superscript ``$\rR$'') Jost solutions  which for real $k$  are defined uniquely by their asymptotics 
 \begin{eqnarray}
 \label{Jost}
 \begin{array}{ccc}
 \phi_1^{\rL}(x;k)\to e^{ikx},\qquad & \phi_2^{\rL}(x;k)\to e^{-ikx} & \mbox{\quad at $x\to-\infty$,} 
 \\[2mm]
 \phi_1^{\rR}(x;k)\to e^{ikx}, \qquad & \phi_2^{\rR}(x;k)\to e^{-ikx} & \mbox{\quad at $x\to+\infty$.} 
 \end{array}
 \end{eqnarray}
 One can also introduce the $2\times 2$ transfer matrix $M = M(k)$ which connects the left and right Jost solutions as follows:
 \begin{eqnarray}
 \label{M}
 \phi^\rL_1 = M_{11}(k)\phi_1^\rR + M_{21}(k)\phi_2^\rR, \qquad \phi^\rL_2 = M_{12}(k)\phi_1^\rR + M_{22}(k)\phi_2^\rR.
 \end{eqnarray} 
 Any solution $\psi(x)$ of (\ref{eq:Sch}) with  real $k$  is a linear combination of left and right Jost solutions with some coefficients $a^{L,R}$, $b^{L,R}$:
 \begin{equation}
 \psi(x) = a^\rL \phi_1^\rL + b^\rL \phi_2^\rL = a^\rR \phi_1^\rR + b^\rR \phi_2^\rR,
 \end{equation}
 The relation among the coefficients is given by the transfer matrix
 \begin{equation}
 \left(\begin{array}{c}
 a^\rR\\b_\rR
 \end{array}\right) = \left(\begin{array}{cc}
 M_{11}&M_{12}\\
 M_{21}&M_{22}
 \end{array}\right) \left(\begin{array}{c}
 a^\rL\\b_\rL
 \end{array}\right).
 \end{equation}
 Thus a positive (negative) zero $k_\star$ of matrix element $M_{22}$ corresponds to a spectral singularity describing the laser (CPA) solution. Conversely, a positive (negative) zero  $k_\star$ of $M_{11}$  describes a CPA (laser).  
 
 While the definition of  transfer matrix introduced above is valid for any scattering potential $U(x)$, the peculiar  property (\ref{eq:pseudo}) imposes   additional relations among the transfer matrix elements. Indeed, let us now apply the $\eta$ operator each Jost solution in (\ref{Jost}).   Using that  for large $x$ the action of   $\eta$   can be approximated by  $\lim_{x\to\pm\infty } \eta = \partial_{x} \mp  ik_0$ and also using the fact that the Jost solutions with real $k$ are defined uniquely by their asymptotic behaviours, we deduce the following relations between the Jost solutions:
 \begin{eqnarray}
 \begin{array}{ll}
(\eta\phi_1^\rL)^* = -i(k+k_0)\phi_2^\rL, & (\eta\phi_1^\rR)^* = -i(k-k_0)\phi_2^\rR,\\[4mm]
(\eta\phi_2^\rL)^* = i(k-k_0)\phi_1^\rL, & (\eta\phi_2^\rR)^* = i(k+k_0)\phi_1^\rR.
\end{array}
\end{eqnarray}
Using these relations together with  (\ref{M}), we obtain that   for all real $k$ the transfer matrix elements are connected as 
 \begin{equation}
 \label{eq:MMMM}
 \begin{array}{c}
 \displaystyle
 M_{11}(k) = M_{22}^*(k)\frac{k+k_0}{k-k_0}, \quad M_{22}(k) = M_{11}^*(k)\frac{k-k_0}{k_0 + k},\\[6mm]
 \displaystyle
 M_{12}(k) = -M_{21}^*(k), \quad M_{21}(k) = -M_{12}^*(k).
 \end{array}
 \end{equation}
 
 To be specific, below we consider the case $k_0>0$. Then we readily conclude that  $M_{22}(k_0) = 0$, which  corresponds to a laser emitting at wavenumber $k=k_0$. This laser solution obviously recovers the already known exact   solution $\psi_0(x)$ given by the explicit expression (\ref{eq:sol_SS}).
 If at some $k_\star$ the two conditions $M_{22}(\ks) = 0$ and $M_{11}(\ks) = 0$ are verified simultaneously. This situation is typical for, although not limited to, $\PT-$symmetric systems~\cite{Longhi10} and such SSs at $\ks$ are sometimes referred to as self-dual~\cite{Most12}. Now one can verify the  spectral singularity at $k_0$ determined by the solution (\ref{eq:sol_SS}) is generically {\em non}-self-dual, i.e. in a general case $M_{11}(k_0)$ is nonzero. Indeed, evaluating $M_{11}(k_0)$ according to 
 	 L'H\^opital's rule, we obtain
 	\begin{equation}
 	M_{11}(k_0) = 2k_0 \frac{d M_{22}^*(k_0)}{dk_0}.
 	\end{equation} 
 	Thus if the spectral singularity $k_0$ is of the first order, i.e., if $M_{22}(k_0)=0$ and $dM_{22}(k_0)/dk\ne 0$,   then $M_{11}(k_0)\neq 0$ and the SS $k_0$ is non-self-dual.  This leads us to the following  necessary and sufficient condition: {\em the SS $k_0$ corresponding to solution (\ref{eq:sol_SS}) is self-dual, iff it is a zero of  the second or higher order of one of diagonal elements of the transfer matrix, i.e., iff the condition one of the conditions $M_{22}(k_0)=dM_{22}(k_0)/dk=0$ or $M_{11}(k_0)=dM_{11}(k_0)/dk=0$ is verified.}
 	  	
 	  	The latter observation is particularly curious in view of the fact that if there is a SS $\ks$ different from $\pm k_0$, it is always self-dual. Indeed it follows from (\ref{eq:MMMM}) that if $M_{22}(k_\star)=0$ with $k_\star\neq \pm k_0$,  then $M_{22}(k_\star)=0$, as well. Thus, for the class of complex potentials under consideration an interesting situation is possible when the potential features an ordinary (non-self-dual) SS $k_0$ corresponding to the exact solution (\ref{eq:sol_SS}) and, at the same time, has a self-dual SS  at  $k_\star \ne k_0$.

 	
 



A representative feature of self-dual SSs for  potentials (\ref{Wadati}) with real $w(x)$  is that the amplitudes of the corresponding CPA and laser solutions, denoted below by $\psi_\mathrm{las}(x)$ and $\psi_\mathrm{CPA}(x)$, respectively, are connected by a simple algebraic relation. Indeed, let  $\ks\ne \pm k_0$ be a self-dual SS. Without loss of generality we can  normalize those solutions  such that $|\psi_\mathrm{las}(x)|^2\to \ri^2$ and  $|\psi_\mathrm{CPA}(x)|^2\to \ri^2$ as $x\to\pm\infty$, where $\rho_\star$ is a positive constant. Then one can show that  the laser and CPA  amplitudes are related as 
 \begin{equation}
 \label{eq:cons}
 |\psi_\mathrm{CPA}|^2 +  |\psi_\mathrm{las}|^2  + \frac{k_0}{\ks} (|\psi_\mathrm{las}|^2 - |\psi_\mathrm{CPA}|^2) = 2\ri^2.
 \end{equation}
 In order to prove this identity, let us  again return from  the scattering problem in the form (\ref{eq:Sch}),   to  the  Schr\"odinger equation (\ref{SE}) with potential (\ref{Wadati}).  It is straightforward to check that  \emph{any} solution $\psi(x)$  of  equation (\ref{SE})  with real $k^2$ satisfies the identity (``the conservation law'') \cite{KZ14} 
 \begin{equation}
 \label{eq:I}
 |\eta \psi|^2 + (k^2-k_0^2) |\psi|^2  = \textrm{const}.
 \end{equation}
 Let $\psi$   in (\ref{eq:I})  be  the laser solution $\psi_\mathrm{las}(x)$ and $k$ in (\ref{eq:I}) be the corresponding wavenumber $\ks$. Due to the symmetry (\ref{eq:pseudo}) a new function $\phi(x) = (\eta \psi_\mathrm{las})^*$ is also a solution of  Schr\"odinger equation (\ref{SE}) with the same $\ks$.   Asymptotic behavior of $\phi(x)$   indicates that  the latter is a CPA solution, and therefore  it  is proportional to the  normalized CPA solution $\psi_\mathrm{CPA}(x)$.  Comparing  the amplitudes of  both solutions,   we obtain   
 \begin{equation}
 \label{eq:aux}
 |\psi_\mathrm{CPA}|^2 = \frac{|\eta \psi_\mathrm{las}|^2}{(\ks-k_0)^2}.
 \end{equation}
 Combining (\ref{eq:I}) and (\ref{eq:aux}), we obtain  (\ref{eq:cons}).

 \subsection{Numerical study of spectral singularities}
 
Now we present several numerical examples illustrating the discussed self-dual SSs. To this end we consider $w(x)$ to be an odd real-valued function which satisfies the boundary conditions (\ref{asympt_w}). The respective potential $U(x)$ in the form (\ref{Wadati}) is an even function (i.e., it is not $\PT$ symmetric). It is clear that if $w(x)$ ia a monotonous function
then no self-dual SSs can exist, because in this case the derivative $w'(x)$ is sign-definite, meaning that there is either only gain (then CPA is impossible) or only loss (then laser is impossible). Therefore, in order to have a self-dual SS and the corresponding CPA-laser, one needs to employ more sophisticated nonmonotonic functions      $w(x)$.
 
 
 As a  representative example, here we consider
 \begin{equation}
 \label{eq:erf}
 w(x) = -k_0\erf (x)\left(1 - \gamma e^{-x^2}\right), \qquad k_0=1,
 \end{equation}
 where $\erf(x)$ is the error function and $\gamma$ is a real parameter that tunes the shape of function $w(x)$. For $\gamma=0$, the potential $U(x)$ generated according to (\ref{Wadati})     corresponds to amplifying media on the whole space.
However, for sufficiently large positive $\gamma$,    the resulting  potential $U(x)$   corresponds to spatially localized absorbing region  sandwiched  between two  amplifying domains [see Fig.~\ref{fig:erf}(a,b) below for  representative plots of the resulting complex potential]. It is interesting to notice that 
the quantity
 \begin{equation}
 \Gamma=w(-\infty)-w(\infty)=\mbox{Im}\,\int_{-\infty}^{\infty}U(x)dx
 \end{equation}
 describes the balance of the total energy either pumped into, $\Gamma>0$, or absorbed by $\Gamma<0$, the system. In our case $\Gamma=2k_0$, i.e., it is fixed  by  $k_0$, and thus does not depend on the value $\gamma$ controlling distribution of the energy in space.  
 
 In contrast to the exact laser solution $\psi_0(x)$ given by (\ref{eq:sol_SS}) and   existing for any $\gamma$ at $k=k_0$,   eventual self-dual SSs can be found only for isolated values of $\gamma$. Obtaining a self-dual SS is reduced to numerical solution of a system of two equations 
 \begin{equation}
 \mathrm{Re}\,M_{22}(k)=0, \qquad \mathrm{Im}\,M_{22}(k)=0,\qquad k\neq k_0,
 \end{equation}
 with respect to two real variables: $\gamma$ and $k$. 
 
  \begin{figure}
 	\centering
 	\includegraphics[width=0.6\columnwidth]{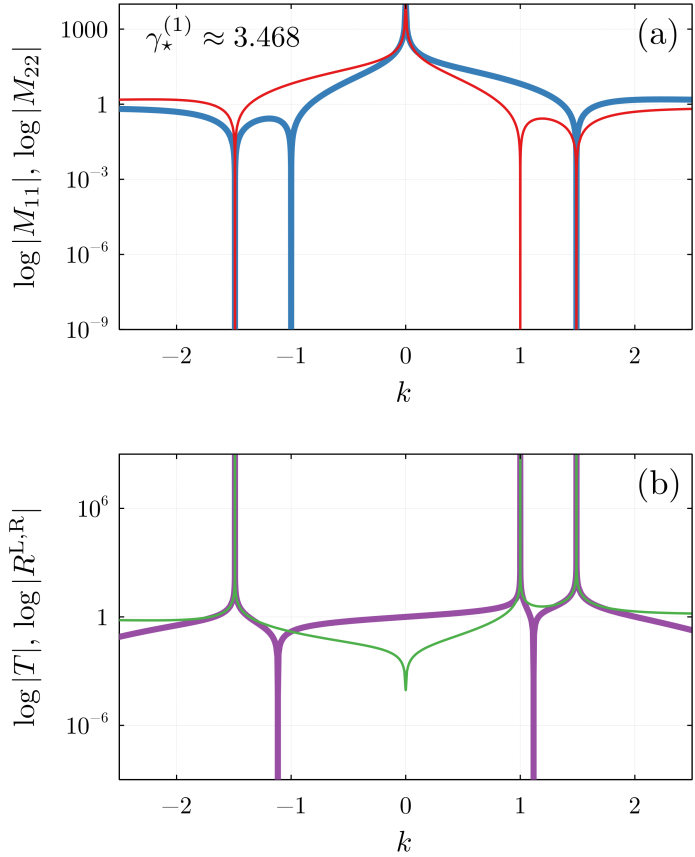}
 	\caption{(a) Logarithmic plots of the absolute values $|M_{11}(k)|$ (bold blue curve) and $|M_{22}(k)|$ (thin red curve) for the potential $U(x)$ generated by function (\ref{eq:erf}) with $\gamma=\gamma_\star^{(1)}\approx 3.468$.  The function $w(x)$ and its derivative, as well as laser and CPA solutions corresponding to the self-dual spectral singularity are shown in Fig.~\ref{fig:erf}(a-f) below. (b) Logarithmic plots of magnitudes of reflection $|R^\rL|=|R^\rR|$ (bold magenta curve) and  transmission $|T|$ (thin green  curve) coefficients computed from the transfer matrix elements plotted in (a).}
 	\label{fig:M11M22}
 \end{figure}
 
Computing numerically the Jost solutions and the elements of the transfer matrix,  we tune   the free parameter $\gamma$ and wavevector $k$ in order to reach a self-dual SS. The smallest positive value of $\gamma$ that enables a self-dual SS is  $\gamma = \gamma_\star^{(1)}\approx 3.468$.
 Logarithmic plots  of  diagonal transfer matrix  elements at $ \gamma_\star^{(1)}$ are shown in figure~\ref{fig:M11M22}(a), where two spectral singularities at different values of $k$ are well visible. The first spectral singularity   corresponds to the zero of $M_{22}(k)$ at $k=k_0=1$ and, equivalently, to the zero of  $M_{11}(k)$ at $k=-k_0=-1$ and   recovers the exact laser solution $\psi_0$  given by (\ref{eq:sol_SS}). Notice that another diagonal element is not zero for this spectral singularity: $M_{11}(k_0) \ne 0$ and $M_{22}(-k_0) \ne 0$. Thus this spectral singularity is not self-dual. The second spectral singularity  occurs  at   $k_\star^{(1)} \approx \pm 1.490$ and is self-dual, since both diagonal elements are zero at this wavevector: $M_{11}(\pm k_\star^{(1)} ) = M_{22}(\pm k_\star^{(1)} )=0$.

\begin{figure}
	\centering
	\includegraphics[width=0.62\columnwidth]{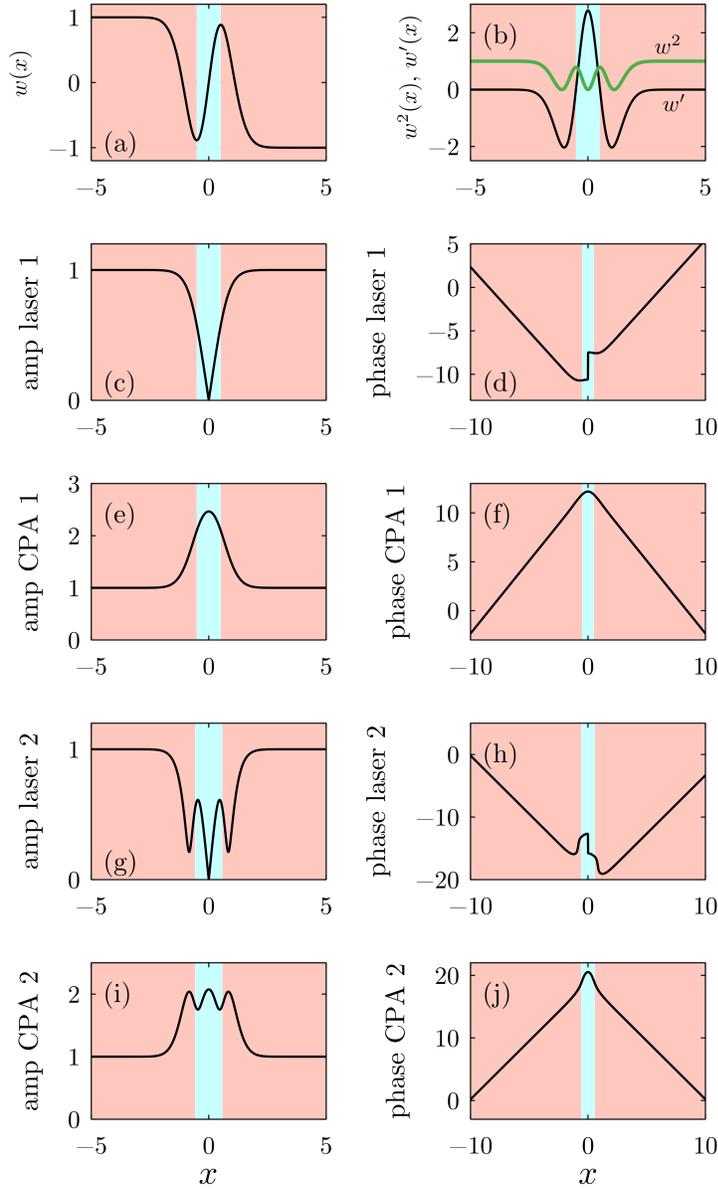}
	\caption{Upper row: plots of $w(x)$ and $w^2(x),\,  w'(x)$ given by (\ref{eq:erf}) with $\gamma = \gamma_\star^{(1)}\approx 3.468$. Spatial domains with energy loss and gain  are filled with  cyan  and  pink   colours, respectively. Second and third rows:  amplitude and phase of coexisting laser and CPA  solutions corresponding to the self-dual SS at $\gamma_\star^{(1)}$ and  $k_\star^{(1)}$. Two bottom rows show the coexisting laser and CPA solutions    for the next self-dual SS  at $\gamma_\star^{(2)}$ and  $k_\star^{(2)}$. }
	\label{fig:erf}
\end{figure} 
 
CPA and laser solutions coexisting at the self-dual SS  are shown in Fig.~\ref{fig:erf}(c,d,e,f). The laser solution is an odd function of $x$  (hence its phase undergoes a $\pi$ jump at $x=0$), whereas the CPA solution is an even function of $x$ and has the intensity peak at $x=0$. 
 Amplitudes of the coexisting laser and CPA solutions shown in  Fig.~\ref{fig:erf}   are connected through identity  (\ref{eq:cons}), where the normalization constant is fixed as $\ri=1$. Interestingly, the amplitude of the CPA solution in the central region is larger than the amplitude of   background  radiation. This  highlights the fact that   enhanced absorption results from the constructive interference of coherent scattered waves in the central, i.e. absorbing domain (notice that similar type of behavior was observed in the recent experiments~\cite{experiment}). Respectively, the laser solution corresponds to the destructive interference in the central region where the solution  amplitude has a node. %

Further, we use transfer matrix in order to evaluate the transmission coefficient $T(k)$ 
and left and right reflection coefficients, $R^\rL(k)$ and   $R^\rR(k)$, using the standard formulas
\begin{equation}
T = \frac{1}{M_{22}}, \qquad R^\rL = -\frac{M_{21}}{M_{22}}, \qquad R^\rR = \frac{M_{12}}{M_{22}}.
\end{equation}
In the case at hand  the left and right transmission coefficients coincide \cite{Most09}, i.e. $T:= T^\rL = T^\rR$.
The amplitudes of obtained scattering coefficients are plotted in figure~\ref{fig:M11M22}(b); notice that in view of relations (\ref{eq:MMMM})  $|R^\rL(k)|\equiv |R^\rR(k)|$, i.e. the left and right reflection coefficients differ only by phases.  The difference between the two coexisting SSs  becomes evident if one compares the behaviour  of the scattering coefficients  for positive and negative values of $k$. For the self-dual SS, all three scattering coefficients diverge at  $ k_\star^{(1)}$ and $ -k_\star^{(1)}$ producing two peaks of infinite height. The ordinary (non-self-dual) SS at $k=\pm k_0$ produces a single peak at $k=k_0$ without a twin in the negative $k$-half-axis. Additionally, in figure we observe that the potential becomes bidirectionally reflectionless at $k\approx \pm 1.116$ where both reflection coefficients vanish simultaneously.

New self-dual SSs can be found for larger values of the parameter $\gamma$ in (\ref{eq:erf}). For instance, upon increasing $\gamma$ the next spectral singularity occurs at $\gamma^{(2)}_\star \approx 11.082$ and $k_\star^{(2)} \approx 1.867$. In this case we again observe that the laser and CPA solutions are odd and even functions of $x$, respectively. Their shapes are illustrated in figure~\ref{fig:erf}(g,h,i,j). Further increase of $\gamma$ leads to the   next spectral singularity at  $\gamma^{(3)}_\star \approx 40.203$ and $k_\star^{(3)} \approx  2.417$. The corresponding solutions (not shown in the paper) bear multiple intensity oscillations and again feature the same parity, with laser and CPA being odd and even functions, respectively. Thus we can conjecture that, apart from the non-self-dual SS $k_{0}$ at any value of $\gamma$,  the potential (\ref{eq:erf}) features a sequence of self-dual SSs with $\gamma_\star^{(1)}<\gamma_\star^{(2)}<\ldots $ and $k_0 <k_\star^{(1)}<k_\star^{(2)}\ldots $. Notice that in the presented example, the energy balance integral $\Gamma$ remains constant, i.e. when increasing gain one also has to increase losses.

We have also computed numerically the entire spectrum of eigenvalues of operator $H$ and observed that the increase of $\gamma$ above each self-dual SS leads to the bifurcation of a new complex conjugate pair of eigenvalues from the continuous spectrum. The spectrum is purely real for $0\leq \gamma\leq \gamma_\star^{(1)}$, but features a single complex conjugate pair for $\gamma_\star^{(1)}< \gamma\leq \gamma_\star^{(2)}$, two complex conjugate pairs for $\gamma_\star^{(2)}< \gamma\leq \gamma_\star^{(3)}$, etc.  This is illustrated in Fig.~\ref{fig:erfeig}, where the entire  spectra of eigenvalues   are illustrated for two values of   $\gamma$. Thus the increase of the parameter $\gamma$ above the first CPA-laser threshold $\gamma_\star^{(1)}$ triggers the phase transition from purely real to complex spectrum (such phase transition was discovered numerically in ~\cite{J17} and described analytically in~\cite{KZ17}).
 

 \begin{figure}
 	\centering
 	\includegraphics[width=0.65\columnwidth]{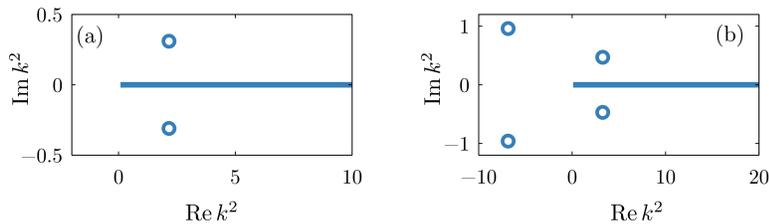}
 	\caption{Spectrum of eigenvalues for the Schr\"odinger operator with the potential (\ref{eq:erf}) for (a) $\gamma=4$ and (b) $\gamma=11.5$. Thick line shows the continuous spectrum and circles indicate the isolated eigenvalues. For $\gamma\leq \gamma_\star^{(1)}\approx 3.468$ the spectrum is purely real and continuous (not shown in the figure).}
 	\label{fig:erfeig}
 \end{figure}

 \section{Discussion and conclusion}
\label{sec:concl}
 
 The main result of this work is that in a quite general physical situation the existence of a SS in the spectrum of a one-dimensional Sch\"odinger operator implies a universal representation of the complex potential, which is given by (\ref{Wadati}). This form of the potential, being a subject of many theoretical studies  in the last decade, is not only a formal algebraic construction but models experimentally feasible potentials, say realized recently in acoustic systems~\cite{experiment}. Complex potentials of this form can be also implemented  in a coherent atomic system driven by laser fields \cite{Chao}.
 
  Functional dependence of the potential (\ref{Wadati}) on only one base function $w(x)$ and on the wavevector $k_0$ at which a SS is observed, allows to 
 engineer 
 complex potentials featuring SSs at a given wavelength. 
 
 We have shown that the corresponding eigenvalue problem has an exact solution (\ref{eq:sol_SS}) which can be either a CPA or laser, both corresponding to real $k_0$. Relaxing the condition of the reality of $k_0$, i.e. considering it in the complex plane, by changing $k_0$ one can transform  the complex potential such that instead of a SS its spectrum can contain bound states, quasi-bound states in continuum, and exceptional points. 

Generically a SS singularity described by the exact solution (\ref{eq:sol_SS}) is simple and non-self-dual. For the particular case when the base function $w(x)$ is real-valued we established that for the exact SS-solution to be a self-dual SS it must be also a second order SS. A numerical example of a potential featuring one simple SS (corresponding to the exact solution) and a set of self-dual SSs was considered in details. We have additionally computed the spectrum of eigenvalues of the corresponding complex potentials and confirmed that a transition through a self-dual spectral singularity generically leads to a bifurcation of a complex-conjugate pair of discrete eigenvalues from  an interior point of the continuous spectrum.
 
As concluding remarks, we mention two straightforward generalizations of the presented theory. First, one can consider base functions $w(x)$ that approach different values at the infinities: $\lim_{x\to\pm\infty} w(x) = k_0^\pm$. In particular, one can construct  ``one-sided'' spectral singularities, i.e. a ``hybrid'' between  SS at one infinite and bound state at another infinity. However, in this case the potential $U(x)$ is not localized, unless $k_0^+ = \pm k_0^-$.
 
Second, for real-valued functions $w(x)$ the exact solution (\ref{eq:sol_SS}) can be straightforwardly generalized   on the nonlinear case. Indeed, incorporating in our model the cubic (Kerr) nonlinearity,  instead of (\ref{eq:Sch}) we obtain the nonlinear eigenvalue problem
 \begin{equation}
 \label{NLS}
 -\psi_{xx} + (-w^2 -iw_x + k_0^2)\psi + \sigma |\psi|^2\psi= k^2\psi.
 \end{equation}
 General properties of this equation  were discussed in \cite{KZ14} and solution of the type (\ref{eq:sol_SS}), although not a SS-solution, was addressed in~\cite{Musslimani}. 
 The exact SS-solution $\psi_0(x)$ of (\ref{NLS}) can be found in the same form as in Eq.~ (\ref{eq:sol_SS}), with the simple shift of the nonlinear eigenvalue:  $k^2 = k_0^2 + \sigma |\rho|^2$. Since  the found nonlinear solution features purely outgoing or purely incoming (depending on the sign of $k_0$) wave boundary conditions,  it paves the way towards the implementation  of a CPA for  nonlinear waves which was experimentally implemented in a Bose-Einstein condensate~\cite{nonlin} and theoretically predicted for arrays of optical waveguides~\cite{ZOK}.




\section*{Acknowledgments}  
 	V.V.K. is grateful to Stefan Rotter for fruitful discussions. Work of D.A.Z. is supported by Russian Foundation for Basic Research (RFBR)
 	(19-02-00193$\backslash$19)

\end{document}